\documentclass[journal,onecolumn]{IEEEtran}

\usepackage[utf8]{inputenc} 
\usepackage[T1]{fontenc}
\usepackage{url}              
\usepackage{cite}             

\usepackage[cmex10]{amsmath}  
\usepackage{mleftright}       
\mleftright                   

\usepackage{graphicx}         
\usepackage{booktabs}         

\usepackage{pgf, tikz, color, xcolor, soul}
\usetikzlibrary{shapes.geometric,arrows,positioning}
\tikzstyle{decision} = [diamond, draw, fill=blue!20, 
    text width=5em, text badly centered, node distance=4cm, inner sep=0pt]
\tikzstyle{block} = [rectangle, draw, fill=blue!20,  text centered, rounded corners, minimum height=4em]
\tikzstyle{line} = [draw, -latex']
\tikzstyle{cloud} = [draw, ellipse,fill=red!20, node distance=6.6cm,
    minimum height=1em]
\tikzstyle{algorithm} = [rectangle, draw, fill=green!20,  text centered, rounded corners, minimum height=4em, minimum width =6em]
\tikzstyle{initialization} = [rectangle, draw,   text centered, minimum height=4em, minimum width =6em]
\def\BibTeX{{\rm B\kern-.05em{\sc i\kern-.025em b}\kern-.08em
    T\kern-.1667em\lower.7ex\hbox{E}\kern-.125emX}}
    
    \tikzstyle{block}=[draw, rectangle, minimum height=1cm, text width=1.5cm, text centered, draw=darkgray, font=\small]
\tikzstyle{block_medium}=[draw, rectangle, minimum height=1.5cm, text width=2cm, text centered, draw=darkgray, font=\small]
\tikzstyle{block_large}=[draw, rectangle, minimum height=2cm, text width=2cm, text centered, draw=darkgray, font=\small]
\tikzstyle{line} = [draw, -latex]

\usepackage{amsfonts, amsmath, amssymb, amsthm}
\newtheorem{theorem}{Theorem}
\newtheorem{proposition}{Proposition}
\newtheorem{corollary}{Corollary}

\newtheorem{remark}{Remark}

\newcommand{\y}{\mathbf{y}}
\newcommand{\x}{\mathbf{x}}

\begin{document}

\title{Correcting One Error in Non-Binary Channels with Feedback}

\author{
\IEEEauthorblockN{Ilya Vorobyev\IEEEauthorrefmark{1},
                    Vladimir Lebedev, \IEEEauthorrefmark{2},
                    Alexey Lebedev \IEEEauthorrefmark{2}}\\
\IEEEauthorblockA{
\IEEEauthorrefmark{1}\textit{Institute for Communications Engineering}, 
\textit{Technical University of Munich}, 
ilya.vorobyev@tum.de}\\
\IEEEauthorblockA{\IEEEauthorrefmark{2}
\textit{Kharkevich Institute for Information Transmission Problems, Moscow},
\{lebedev37, al\_lebed95\}@mail.ru}
}
\maketitle

\begin{abstract}
In this paper, the problem of correction of a single error in $q$-ary symmetric channel with noiseless feedback is considered. We propose an algorithm to construct codes with feedback inductively. For all prime power $q$ we prove that two instances of feedback are sufficient
to transmit over the $q$-ary symmetric channel the same number of messages as in the case of complete feedback. Our other contribution is the construction of codes with one-time feedback with the same parameters as Hamming codes for $q$ that is not a prime power. We also construct single-error-correcting codes with one-time feedback of size $q^{n-2}$ for arbitrary $q$ and $n\leq q+1$, which can be seen as an analog for Reed-Solomon codes.
\end{abstract}

\section{Introduction}\label{sec::introdiction}

We consider the problem of correction of a single error in $q$-ary symmetric channel with noiseless feedback. This problem is a generalization of an analogous problem for a binary case, which was investigated in many papers.
It is known that the problem of correcting up to $t$ errors in the binary channel with complete feedback is equivalent to the following problem of combinatorial search. 
The task is to find an element $x \in \mathcal{M}$ with the help of $n$ questions of the form: ``Is it true that the element $x$ belongs to a subset $A$ of the set $\mathcal{M}$?'' Questions are asked sequentially, i.e., each question can depend on the answers to the previous ones. The opponent, who answers the questions, knows $x$ and can lie at most $t$ times. This problem was initially formulated by Renyi~\cite{renyi61}. 
It became popular after Ulam asked to solve it for  $M=10^6$ in his autobiography~\cite{ulam1991adventures}. Optimal strategies have been found for all $M$ in \cite{P87} for $t=1$, in \cite{G90} for $t=2$ and in \cite{D00} for $t=3$. For a linear number of errors $t=\tau n$ the optimal rate in a binary channel with complete feedback was calculated by Berlekamp~\cite{B68} and Zigangirov~\cite{zigangirov1976number}.

The problem of correcting errors in $q$-ary symmetric channel with complete feedback is equivalent to the generalization of Ulam's game, where you can ask to which of $q$ sets the element $x$ belongs. For a linear number of errors the generalization of Berlekamp's results for $q$-ary case was obtained in~\cite{ahlswede2006non}. For a fixed number of errors Bassalygo~\cite{bassalygo2005nonbinary} proved even one-time feedback is enough to transmit asymptotically the same number of messages as for the case with complete feedback. For an arbitrary discrete channel analogous result was obtained in \cite{dumitriu2005two}. The case of a fixed number of errors was also investigated in papers~\cite{cicalese2000perfect}, \cite{cicalese2003quasi}, and~\cite{ahlswede2008searching}. For one error the optimal number of messages was computed by Aigner in~\cite{aigner1996searching}.




\subsection{Our contribution}

In this paper we investigate the codes correcting a single error in a $q$-ary symmetric channel with one-time, two-times, or complete feedback.
We propose an algorithm, called non-binary Double and Delete Algorithm (DADA), which allows for constructing codes with feedback inductively. The most interesting result is the construction of a single-error correcting code for $q$-ary channel ($q$ is a prime power) with two instances of feedback, which has the same number of messages as the optimal code for a channel with complete feedback.
The other contribution is the construction of codes with one-time feedback, whose sizes lie on Hamming bound for $q$ that is not a prime power and $n=\frac{q^k-1}{q-1}$. For arbitrary $q$ and $n\leq q+1$, we present a code of size $q^{n-2}$ with one-time feedback. Note for a case without feedback and prime power $q$ Reed-Solomon codes have the same parameters. However, for $q$ that is not a prime power such codes are not known.

\subsection{Outline}
The rest of the paper is structured as follows. In Section~\ref{sec::notation} necessary notations and definitions are introduced. In Section~\ref{sec::non-binary BSC with one-time feedback} our results for a symmetric channel with one-time feedback are presented. An algorithm, which allows inductive construction of codes from codes of smaller length for the symmetric channel, is described in Section~\ref{sec::non-binary DADA}. Section~\ref{sec::complete feedback} is devoted to codes with complete feedback. In Section~\ref{sec::two instances of feedback} coding with two instances of feedback is discussed and the main theorem is proved. In Section~\ref{sec::other channel} an example of another $q$-ary channel, for which the developed methods allow to compute the maximal number of messages, is provided.

\section{Notations and Definitions}\label{sec::notation}
We use bold latin letters $\x$, $\y$, to refer to vectors, $[M]$ denotes the set $\{1, 2, \ldots, M\}$.  The prefix of length
$p$ of a vector $\y$ is denoted by $\y_{\overline{p}}$. 
In this paper we consider transmission over a $q$-ary symmetric channel, i.e., the input and output alphabets consist of $q$ symbols $\{0, 1, \ldots, q-1\}$, and each of these symbols can be transmitted incorrectly and received as any other symbol from the alphabet. We consider a channel with a single error, i.e., Hamming distance between transmitted vector $\x$ and received vector $\y$ is at most $1$. The coder can use $s$ instances of feedback in the following manner.  Let the
codeword length~$n$ be divided into $s+1$ parts:
$$
n=n_1+n_2+\ldots+n_{s+1}.
$$
The coder encodes a message $m\in [M]$ as follows. The first $n_1$ transmitted symbols
$x_1,\ldots, x_{n_1}$ depend on the message $m$ only. After $N_{i-1}:=n_1+\ldots
+n_{i-1}$, $i\ge 2$, symbols are transmitted, the encoder has values of the received
symbols $\y^{}_{\overline{N_{i-1}\!}}$ from the feedback channel. The encoder
sends the $i$th block of $n_i$ symbols, which is a function of the message $m$ and of
the symbols $\y^{}_{\overline{N_{i-1}\!}}$ received by the encoder. The case of
$s=0$ corresponds to a channel without feedback, and the case of $s=n-1$, to a
channel with complete feedback.

Fix the encoding function. Define the cloud $B(m)$ for a message $m$ to be the set of
sequences $\y$ that can be obtained at the output of the channel with at
most one error during the transmission associated with this message. We refer to the collection
of disjoint clouds $B(m)$, $m\in [M]$, as a single-error correcting
code~$\mathcal{C}$. Sequences that do not belong to any cloud will be
referred to as free sequences (or free points) and will be denoted by~$\mathcal{F}(\mathcal{C})$. Codes
that do not use feedback will be called nonadaptive.

Let us describe the structure of a cloud for a non-binary channel with a
single error and complete feedback. Every cloud $B(m)$ contains a sequence $\y$
which will be transmitted if there are no errors in the channel. We call it a root sequence. For any coordinate $i$, the cloud contains $q-1$ sequences $\y(i, j)$, $j=1, \ldots, q-1$, which
coincide with $\y$ in the first $i-1$ positions, have symbols in $i$th
position, which are different from each other and different from $\y_i$, and has arbitrary symbols in all other positions. Hence it is seen that
each cloud consists of at least $(q-1)n+1$ sequences. In particular, this yields the well-known
Hamming bound on the maximum number of transmitted messages with complete feedback.

\section{Non-binary symmetric channel with one-time feedback}\label{sec::non-binary BSC with one-time feedback}
Consider an arbitrary discrete memoryless channel with input and output alphabets $\{0, 1, \ldots, q-1\}$.
Divide the codeword length $n$ into two parts, $n_1$ and $n_2$,
with $n=n_1+n_2$. Define a bipartite graph $H=(U\sqcup V, E)$ as follows. The
left- and right-hand parts $U$ and $V$ consist of $q^{n_1}$ vertices corresponding to the sets of
input and output sequences. Vertices $u$ and $v$ are connected by an edge if the
sequence corresponding to $v$ can be obtained from the sequence corresponding to~$u$
as a result of a single error. Note that vertices corresponding to
identical sequences are not connected (this corresponds to the case with no error). In other words, this graph describes possible errors in the channel.

Recall the theorem proved in~\cite{vorobyev2022correcting}.
\begin{theorem}\label{th::main}
Let a $q$-ary channel and its associated bipartite graph $H=(U\sqcup V, E)$ be given. A strategy allowing to transmit
\begin{equation}\label{number of transmitted words}
M=\sum\limits_{u\in U}M(u)
\end{equation}
messages over this channel with a single error exists if and only if there exists a family of single-error-correcting codes
$C(u)$ of length~$n_2$ and cardinality $M(u)$ with $F(u)$ free points that satisfy
the condition
\begin{equation}\label{constraints}
\sum\limits_{u:\: (u,v)\in E} M(u)\le F(v)
\end{equation}
for any $v\in V$.
\end{theorem}



Theorem~\ref{th::main} allows to prove the following
\begin{theorem}
\label{th::q-ary symmetric channel with one-time feedback}
Let $n_1, n_2$, $n_1 + n_2=n$. If there exists a single-error correcting code of length $n_2$ and size
$
\left\lfloor
\frac{q^{n_2}}{(q-1)n+1}
\right\rfloor,
$
then it is possible to transmit
$$
M_1(n)=q^{n_1}
\cdot \left\lfloor
\frac{q^{n_2}}{(q-1)n+1}
\right\rfloor
$$ 
messages over a non-binary channel with one error and one-time feedback.
\end{theorem}

\begin{proof}[Proof of Theorem~\ref{th::q-ary symmetric channel with one-time feedback}]
We apply Theorem~\ref{th::main} where as codes $C(u)$ we take the code of size 
$$
M(u)=M=\left\lfloor
\frac{q^{n_2}}{(q-1)n+1}
\right\rfloor.
$$

To satisfy conditions~\eqref{constraints} of Theorem~\ref{th::main} the following inequality
$$
(q-1){n_1}M\leq q^{n_2}-(n_2+1)M
$$
must be true. It is true, since it is equivalent to
$$
M\leq \frac{q^{n_2}}{{n_1}(q-1)+(n_2+1)}=\frac{q^{n_2}}{(q-1)n+1},
$$
which holds by definition of $M$.

By using Theorem~\ref{th::main} we conclude that it is possible to transmit $q^{n_1}M$ messages.

\end{proof}



\begin{corollary}\label{cor::hamming with feedback}
Let $n=\frac{q^k-1}{q-1}$, $k\geq 2$. Then the maximal number of messages that can be transmitted over a $q$-ary symmetric channel with a single error and one-time feedback is equal to
$$
M_1=q^{n-k}.
$$
\end{corollary}
\begin{proof}[Proof of Corollary~\ref{cor::hamming with feedback}]
    We apply Theorem~\ref{th::q-ary symmetric channel with one-time feedback} for $n_2=k$, $n_1=n-k$. Code of size $\left\lfloor
\frac{q^{n_2}}{(q-1)n+1}
\right\rfloor=1$ always exists, therefore we can transmit
$
q^{n_1}=q^{n-k}
$
messages. This number of messages is maximal since it lies on Hamming bound.
\end{proof}
\begin{remark} 
When $q$ is prime power and $n=\frac{q^k-1}{q-1}$ it is possible to transmit $q^{n-k}$ messages even without feedback by using Hamming codes. Our
Corollary~\ref{cor::hamming with feedback} shows that with one-time feedback such number of messages can be transmitted also for $q$ that are not prime powers.
\end{remark}
There exist other pairs $n$ and $q$, for which codes from Theorem~\ref{th::q-ary symmetric channel with one-time feedback} achieve Hamming bound. For example, for $q=6$, $n=97$, $n_1=92$, $n_2=5$, conditions of Theorem~\ref{th::q-ary symmetric channel with one-time feedback} are satisfied due to Gilbert-Varshamov bound. Then $q^{n_2}$ is divisible by $(q-1)n+1$ and $M_1(n)=\frac{q^n}{(q-1)n+1}$, i.e., the number of messages coincides with a Hamming bound.

In the following theorem we provide codes with one-time feedback, which can be seen as an analog to Reed-Solomon codes for arbitrary $q$.
\begin{theorem}\label{th::n less q+1}
	For an arbitrary $q$ and $n\leq q+1$ there exists a code with one-time feedback of size $q^{n-2}$.
\end{theorem}
Note that even with complete feedback it is impossible to transmit more than $q^{n-2}$ messages~\cite{aigner1996searching}.
\begin{proof}[Proof of Theorem~\ref{th::n less q+1}]
	We are going to use Theorem~\ref{th::main}.
	Let $n_1=n-2$, $n_2=2$. As codes $C(u)$, we use codes of size 1. Let us check the condition~\eqref{constraints}.
	$$
	(q-1)\cdot n_1\leq q^2-(1+2(q-1)),
	$$
	which is equivalent to
	$$
	n\leq q+1.
	$$
	The number of transmitted messages is equal to the sum of sizes of codes $C(u)$, i.e., $q^{n-2}.$
\end{proof}

 \section{Non-binary DADA}\label{sec::non-binary DADA}
 In paper~\cite{vorobyev2022correcting} the authors proposed DADA (Double and Delete Algorithm), which constructs codes with feedback for binary symmetric channel inductively from codes with smaller lengths. Here we generalize this approach for $q$-ary channel.

Let us briefly describe non-binary DADA algorithm.

Assume that we have a set of $M(n-1)$ clouds for the length $n-1$. 

\begin{enumerate}
\item From each cloud of sequences of length $n-1$, we construct $q$ sets of sequences of length $n$. To construct the $i$-th set we  add the prefix $x=i-1$ to each sequence of the cloud. Call these sets incomplete clouds. 
\item We add sequences, which do not  belong to any cloud or incomplete cloud, to incomplete clouds to transform them into clouds for length $n$.
\item While there are incomplete clouds we delete some of them and use their sequences to transform some incomplete clouds into clouds.
\end{enumerate}

The details of how we choose which incomplete clouds to delete are given in the proof of Theorem~\ref{th::non-binary DADA}.

\begin{theorem}\label{th::non-binary DADA}
Let $n>q$. Assume that there exists a code of length $n-1$ with $s-1$ instances of feedback of size $M(n-1)$, each cloud of which has cardinality $1+(n-1)(q-1)$, correcting one error in $q$-ary symmetric channel, $s=1,\ldots, n-1$.
Let 
\begin{align}
    U_q(n) &= q\left\lfloor\frac{q^n}{q(1+n(q-1))}\right\rfloor,\\
    p_q(n) &= (1+n(q-1))(U_q(n)+q)-q^n,\\
    H(n) &= \left\lfloor\frac{q^n}{1+n(q-1)}\right\rfloor.
\end{align}
Then non-binary DADA constructs a code $C^q_s(n)$ of length $n$ with $s$ instances of feedback and size $M(n)$, where
$$
M(n)=
\begin{cases}
qM(n-1), &\;\;\text{ if } qM(n-1)< H(n);\\
U_q(n), &\begin{array}{l}\text{ if } qM(n-1)\geq H(n)\\ \text{ and } p_q(n)\geq q^2;
\end{array}\\
U_q(n)+q-r,&\begin{array}{l}\text{ if }qM(n-1)\geq H(n)\\\text{ and }p_q(n)=qr.
\end{array}
\end{cases}
$$
\end{theorem}
\begin{remark}
    $U_q(n)$ is the maximal number, which is not greater than Hamming bound and divisible by $q$, 
    $p_q(n)$ is the number of points, needed for $U_q(n)+q$ clouds minus the total number of points, $H(n)$ is an upper Hamming bound on the size of a code of length $n$. Theorem says that the size of the code is multiplied by $q$ if the new size is less than Hamming bound; otherwise, the new size is equal to or slightly bigger than the $U_q(n)$.
\end{remark}
\begin{proof}[Proof of Theorem~\ref{th::non-binary DADA}]
Let us describe the structure of a cloud for the non-binary symmetric channel. Each cloud contains a root sequence $\x$ together with other $(q-1)n$ sequences. For each coordinate $i$ from $1$ to $n$ and each symbol of the alphabet $z\ne x_i$, there is a sequence $\y$ in the cloud, such that $\x_{\overline{i-1}}=\y_{\overline{i-1}}$ and $y_i=z$.

From each cloud of sequences of the length $n-1$ we construct $q$ sets of sequences of length $n$ by adding to each sequence the prefix $i - 1$ for $i$-th set, $i=1, 2, \ldots, q$. To make a cloud from $i$-th set  it is enough to add any $q-1$ sequences starting from all $j\ne i - 1$. 
Throughout the proof we will refer to sequences, which do not belong to any complete or incomplete cloud, as free points.
Free points will be used to create clouds from incomplete clouds. We take $q(q-1)$ free points, $q-1$ of which are beginning from the symbol $s\in\{0, 1, \ldots, q-1\}$. Then we use these free points to turn $q$ incomplete clouds into clouds. We note that after this operation for any two symbols $i$ and $j$ the numbers of incomplete clouds, which sequences start with $i$ and $j$, are the same. Also, the numbers of free sequences beginning with $i$ and $j$ are the same. We repeat the described operation while it is possible. If we cannot perform the operation because there are no more incomplete clouds, then we have obtained $qM(n-1)$ clouds and some number of free points, where $M(n-1)$ is the size of $C_{k-1}^q(n-1)$. In this case, the algorithm is completed.

Otherwise, at the end of this procedure, we obtain some number of clouds, incomplete clouds, and less than $q(q-1)$ free points.

In what follows we take $q$ incomplete clouds, all sequences of $i$-th cloud start with a symbol $i-1$. We delete these incomplete clouds and transform all their sequences into free points. Such operation gives $q(1+(q-1)(n-1))$ free points, which are spent to create clouds from incomplete clouds in the same manner as in the first part of the algorithm. One can see that the amount of free points, obtained from $q$ incomplete clouds, is enough to transform at $q(n-1)$ or $qn$ incomplete clouds into clouds. So, if the number of incomplete clouds is at least $q(n+1)$, then we can perform this operation and obtain less than $q(q-1)$ free points at the end. We repeat this operation until the number of incomplete clouds becomes less $q(n+1)$.

Denote the number of remaining incomplete clouds and free points as $qm$ and $q\ell$ respectively, $m\leq n$, $\ell < q - 1$. If $m=0$ then the algorithm is completed. If $m=n$ then we can turn $q$ incomplete clouds into free points and use these free points to transform the res $q(n-1)$ incomplete clouds into clouds. In both cases, the number of free points is at most $q(l+1)\leq q(q-1)<1+(q-1)n$, i.e., less than the volume of one cloud, which means that Hamming bound is achieved and the obtained number of clouds is optimal.

Let $m<n$. Deletion of one incomplete cloud, all sequences of which start with a symbol $i$, gives enough free points starting with $i$ to transform all remaining incomplete clouds into clouds. 

If $m\geq 2$, then we delete $q$ incomplete clouds and use the obtained free points to transform the remaining incomplete clouds into clouds. The number of remaining free points is equal
\begin{multline}
q\ell+qm(1+(q-1)(n-1))-q(m-1)(1+(q-1)n)\\=q(\ell-m(q-1)+1+(q-1)n)\leq q(q-1)(n-1),
\end{multline}
which is less than the volume of $q$ clouds. Therefore, the number of obtained clouds is equal to the maximal number, which is not greater than Hamming bound and divided by $q$, i.e., $U_q(n)$.

Moreover, the number $p_q(n)$ of points, which we lack to have enough points for $q$ more clouds, is at least
$$
q(1+n(q-1))-q(q-1)(n-1)=q^2.$$

Consider the last case $m=1$. We transform $r$ incomplete clouds, which sequences are starting with symbols $1, 2, \ldots, r$, into free points, which are used to create clouds from the remaining incomplete clouds. It is possible if we have free sequences, starting with $r+i$, $i=1, 2, \ldots, q - r$, i.e., if inequality $\ell\geq q-r-1$ holds. The parameter $r$ should be as small as possible, so we take $r=q-1-\ell$.

The number of free points at the end is equal to
\begin{multline*}
    q\ell + (1+(n-1)(q-1))q-(q-r)(1+n(q-1))\\
    =q(q-1)-qr-q(q-1)+r(1+n(q-1))=r(n-1)(q-1),
\end{multline*}
and the number $p_q(n)$ is equal to $qr$.

Theorem~\ref{th::non-binary DADA} is proved.
\end{proof}

\section{Complete Feedback}\label{sec::complete feedback}
In the case of complete feedback the optimal number of messages, which can be transmitted over a $q$-ary symmetric channel, was computed by Aigner in ~\cite{aigner1996searching}. We provide simple proof for this result.

\begin{theorem}\label{th::simpleProofCompleteFeedbackQary}
Let $n\geq q+1$ and
\begin{align*}
U_q(n) &= q\left\lfloor\frac{q^n}{q(1+n(q-1))}\right\rfloor,\\
p_q(n) &= (1+n(q-1))(U(n)+q)-q^n.
\end{align*}
Then it is possible to transmit $M_{cf}(n)$ messages over the $q$-ary symmetric channel with a single error and complete feedback, where
\begin{equation}\label{eq::binary complete feedback}
M_{cf}(n)=
\begin{cases}
U_q(n), &\text{ if }  p_q(n)\geq q^2;\\
U_q(n)+q-r,&\text{ if } p_q(n)=qr.
\end{cases}
\end{equation}
Moreover, this number of messages is optimal.
\end{theorem}

The following proposition will be used in the proof of the theorem, but it is also interesting by itself.
\begin{proposition}\label{prop::Ham+1}
    Let $q$ be a prime power and $n=\frac{q^{k}-1}{q-1}+1$. Then it is possible to transmit $M_{cf}(n)$ messages over a $q$-ary symmetric channel with a single error and one-time feedback.
\end{proposition}
\begin{proof}[Proof of Proposition~\ref{prop::Ham+1}]
    Since $q$ is a prime power there exists a $q$-ary Hamming code of length $n-1=\frac{q^{k}-1}{q-1}$ and size $q^{n-1-k}$.
    We use non-binary DADA to construct a code of length $n$ with one-time feedback from the Hamming code of length $n-1$.
    Note that
    $$
    q\cdot q^{n-1-k}>\frac{q^n}{1+n(q-1)},$$
    hence the size of a new code equals $U_q(n)$ or $U_q(n)+q-r$ depending on $p_q(n)$, i.e., it is equal to $M_{cf}(n)$.
\end{proof}

\begin{proof}[Proof of Theorem~\ref{th::simpleProofCompleteFeedbackQary}]
We are not going to prove optimality, since this part of Aigner's results is already quite simple.

For the case $q=2$ the theorem has already been proved in~\cite{vorobyev2022correcting}, so we consider $q>2$.

We will prove the theorem by induction.
From Theorem~\ref{th::n less q+1} we know that for $n=q+1$ the optimal number of messages equals $q^{n-2}$. We note, that Eq.~\eqref{eq::binary complete feedback} gives the same value. Also from Proposition~\ref{prop::Ham+1} we know that for prime power $q$ the optimal number of messages can be transmitted with one-time feedback for $n=q+2$.

Assume that for length $n-1$, $n \geq q+3$ for prime powers $q$ and $n\geq q+2$ for others $q$, Eq.~\eqref{eq::binary complete feedback} is proved, i.e., there exists a set of $M_{cf}(n-1)$ clouds. We use non-binary DADA to construct a code of length $n$ from this set.

We want to show that 
\begin{equation}\label{eq::condition for DADA in simple proof}
qM_{cf}(n-1)>\frac{q^n}{1+n(q-1)}
\end{equation}

Since
\begin{align*}
qM_{cf}(n-1)\geq q^2\left(\frac{q^{n-1}}{q(1+(n-1)(q-1))}-1\right)\\
=\frac{q^n}{1+(n-1)(q-1)}-q^2
\end{align*}
it is enough to verify
$$
q^n(q-1)>q^2(1+n(q-1))(1+(n-1)(q-1).
$$

It is easy to check that $q^{n-3}\geq n^2$ for $q=3$, $n\geq n+3$ and for $q\geq 4$, $n\geq q+2$.
Therefore, 
\begin{multline*}
    q^n(q-1)\geq q^3(q-1)n^2\\>q^2(1+n(q-1))(1+(n-1)(q-1),
\end{multline*}
i.e., the condition~\eqref{eq::condition for DADA in simple proof} holds. It means that after applying DADA we obtain a code of size 
$U_q(n)$ or $U_q(n)+q-r$ depending on $p_q(n)$, i.e., its size is equal to $M_{cf}(n)$.
\end{proof}

\section{Two instances of feedback}\label{sec::two instances of feedback}
In this Section, we are going to prove our main theorem and show that Aigner's result can be achieved with only two instances of feedback.
\begin{theorem}\label{th::two instances of feedback}
Let $q$ be a prime power and $n\geq q+1$. Then it is possible to transmit $M_{cf}(n)$ messages over a $q$-ary symmetric channel with a single error and two instances of feedback.
\end{theorem}
This theorem can be seen as a generalization of Theorem~4 from~\cite{vorobyev2022correcting}, where the same result was proved for a binary channel.

\begin{proof}[Proof of Theorem~\ref{th::two instances of feedback}]
For $n=\frac{q^k-1}{q-1}$ and $n=\frac{q^k-1}{q-1} + 1$ the desired result can be achieved without feedback and with one-time feedback correspondingly by using Hamming codes and Proposition~\ref{prop::Ham+1}.

We choose $n_2$ as a maximal number of the form $\frac{q^k-1}{q-1}$, which is less than $n-1$, $n_1=n-1-n_2\geq 1$. For such $n_2$ and prime power $q$ there exists a Hamming code, therefore, we can apply Theorem~\ref{th::q-ary symmetric channel with one-time feedback} to construct the code $C_1(n-1)$ with one-time feedback of size 
$$
q^{n_1}\cdot 
\left\lfloor
\frac{q^{n_2}}{(q-1)n+1}
\right\rfloor.
$$

Then we use a non-binary DADA to construct a code $C_2(n)$ from a code $C_1(n-1)$. 

To finish the proof we need we need to show that
\begin{equation}\label{eq::cond for 2 instances of feedback}
   q^{n_1+1}
   \left\lfloor
\frac{q^{n_2}}{(q-1)(n-1)+1}
\right\rfloor>\frac{q^n}{(q-1)n+1} 
\end{equation}
After dividing both sides by $q^{n_1+1}$ and using inequality $\lfloor x\rfloor>x-1$, it is enough to check the following
$$
\frac{q^{n_2}}{(q-1)(n-1)+1}-1>
\frac{q^{n_2}}{(q-1)n+1}, 
$$
which is equivalent to
$$
q^{n_2}(q-1)>(1+(q-1)(n-1))(1+(q-1)n).
$$
Recall that $n_2$ is the maximal number of the form $\frac{q^k-1}{q-1}$, which is less than $n-1$, hence $n\leq qn_2+2$.
Since
\begin{multline*}
	(1+(q-1)(n-1))(1+(q-1)n)\\<(q-1)^2n^2\leq (q-1)^2(qn_2+2)^2, 
\end{multline*}
it is enough to verify that
$$
q^{n_2}\geq(q-1)(qn_2+2)^2
$$
for $n_2\geq q+1$.
It is easy to verify that for $q\geq 4$ this inequality holds. For $q=3$ this inequality is true for $n_2\geq 13$, i.e., for $n\geq 15$.
For $q=3$ and $n\in [6, 14]$ we manually check the existence codes with one-time feedback of length $n-1$ and size $M_1(n-1)$ such $q\cdot M_1(n-1)\geq \frac{q^n}{(q-1)n+1}$.
    
\end{proof}
\begin{remark}
    For the case of $q$ that is not a prime power, we were able to prove that $M_{cf}(n)$ messages can be transmitted with two instances of feedback only for big enough $n>n_0(q)$. We suspect that it can be proved for all $q$ and $n\geq q+1$.
\end{remark}

\section{Other Q-ary Channels}\label{sec::other channel}
In this Section, we provide an example of another $q$-ary channel with complete feedback, for which we can compute the optimal number of messages. In this channel $0$ can be changed into $1$, and $1$ can be changed into $0$, all other symbols are always transmitted without errors. The error graph for this channel is depicted in Figure~\ref{fig2}.

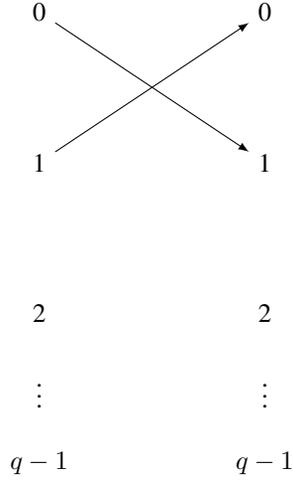
\begin{figure}[h]
\centering
	\begin{tikzpicture}
	\node (A) at (0,0) {0};
	\node (B) at (3,0) {0};
	\node (C) at (0,-2) {1};
	\node (D) at (3,-2) {1};
        \node (E) at (0,-4) {2};
	\node (F) at (3,-4) {2};
 \node (E) at (0,-5) {$\vdots$};
	\node (F) at (3,-5) {$\vdots$};
	\node (E) at (0,-6) {$q-1$};
	\node (F) at (3,-6) {$q-1$};
	\path[line] (C) -- (B);
	\path[line] (A) -- (D);

	\end{tikzpicture}
	
	\caption{Error Graph}\label{fig2}
\end{figure}

For this channel algorithm analogous to non-binary DADA allows us to compute the optimal number of messages under some conditions.
\begin{theorem}\label{th::other channel}
Let $q>2$ and $N$ is the maximal integer, which satisfies inequality $$2(q-2)^{N-1}\geq q^{N-1}.$$
Then for all $n\leq N$ the optimal number of messages, which can be transmitted over the channel with complete feedback depicted in Figure~\ref{fig2}, equals $$M(n)=\frac{q^n+(q-2)^n}{2}.$$
\end{theorem}
\begin{proof}[Proof of Theorem~\ref{th::other channel}]

We will construct a code with complete feedback of length $n$ with $M_1(n)$ clouds of size $1$ and $M_2(n)$ clouds of size $2$, where
\begin{align*}
    M_1(n)&=(q-2)^{n},\\
    M_2(n)&=(q^{n}-(q-2)^{n})/2.
\end{align*}
For $n=1$ we can easily construct a code with $M_1(1)=q-2, M_2(1)=1$. We construct codes for bigger lengths inductively.

Assume that for length $n-1$ we have already constructed a code with 
\begin{align*}
    M_1(n-1)&=(q-2)^{n-1},\\
    M_2(n-1)&=(q^{n-1}-(q-2)^{n-1})/2.
\end{align*}
Then we use the idea similar to non-binary DADA to construct a code of length $n$. By appending one symbol from the set $\{2, 3, \ldots, q-1\}$ at the beginning of each sequence of clouds for length $n-1$, we obtain $(q-2)M_1(n-1)=M_1(n)$ clouds of size $1$ and $(q-2)M_2(n-1)$ clouds of size $2$ for length $n$. After appending $0$ and $1$ we get incomplete clouds. Delete all incomplete clouds of size $2$ to obtain $q^{n-1}-(q-2)^{n-1}$ free points with $0$ at the first position and the same number with $1$ at the first position. We spend all these new free points to transform incomplete clouds of size $1$ into complete clouds of size $2$. Since $q^{n-1}-(q-2)^{n-1}\leq(q-2)^{n-1}$, after this procedure we will have $M_1(n)$ clouds of size $1$, some number of complete clouds of size $2$, and some even number of incomplete clouds of size $1$. By combining two incomplete clouds with starting symbols $0$ and $1$ we obtain a complete cloud of size $2$. We use all incomplete clouds of size $1$ to construct complete clouds of size $2$. In the end, we have $M_1(n)$ clouds of size $1$ and some amount of clouds of size $2$, which completely cover all $q^n$ points. Thus the number of clouds of size $2$ equals $(q^n-(q-2)^n)/2=M_2(n)$.

The constructed code has the maximal possible size, since it contains only clouds of sizes $1$ and $2$, and has the maximal possible number of clouds of size $1$. So, the maximal size of a code equals
$$
M(n)=M_1(n)+M_2(n)=\frac{q^n+(q-2)^n}{2}.$$
\end{proof}

\section*{Acknowledgment}
Ilya Vorobyev was supported by BMBF-NEWCOM, grant number 16KIS1005.

\bibliographystyle{IEEEtran}
\bibliography{mybib}

\end{document}